\documentclass[onecolumn,fleqn,runningheads]{svjour2}
\usepackage{graphicx}
\usepackage{epstopdf}
\usepackage{color}
\usepackage{amsmath}
\usepackage{array}

\usepackage{mathptmx}      
\usepackage{latexsym}
\journalname{China Science Bulletein }
\usepackage{bm}

\newcommand{\mi}{\mathrm{i}}

\newcommand{\dif}{\mathrm{d}}

\DeclareMathAlphabet{\mathsfsl}{OT1}{cmss}{m}{sl}
\newcommand{\tensor}[1]{\mathsfsl{#1}}

\begin{document}

 \title{A Sort of Relation Between a Dissipative Mechanical System and Conservative Ones }
\titlerunning{Relation Between a Dissipative Mechanical System and Conservative Ones}

\author{Tianshu Luo \and Yimu Guo}

\institute{Tianshu Luo, Yimu Guo \at
              Institute of Solid Mechanics, Department of Applied Mechanics, Zhejiang University,\\
Hangzhou, Zhejiang, 310027,  P.R.China          
            \email{ltsmechanic@zju.edu.cn}
           \and
           Yimu Guo \at
              Institute of Solid Mechanics, Department of Applied Mechanics, Zhejiang University,\\
Hangzhou, Zhejiang, 310027,  P.R.China          
            \email{guoyimu@zju.edu.cn}
}

\date{Received: date / Accepted: date}
\maketitle
\begin{abstract}
In this paper we proposed a proposition: for any nonconservative classical mechanical system and any initial condition, there exists a conservative one; 
the two systems share one and only one common phase curve;  the Hamiltonian of the conservative system is the sum of the total energy of the 
nonconservative system on the aforementioned phase curve and a constant depending on the initial condition. Hence, this approach entails 
substituting an infinite number of conservative systems for a dissipative mechanical system corresponding to varied initial conditions. 
One key way we use to demonstrate these viewpoints is that by the Newton-Laplace principle the nonconservative force can be reasonably assumed to be 
equal to a function of a component of generalized coordinates $q_i$ along a phase curve, such that a nonconservative mechanical system can be reformulated 
as countless conservative systems. Utilizing the proposition, one can apply the method of Hamiltonian mechanics or Lagrangian mechanics to dissipative mechanical
 system. The advantage of this approach is that there is no need to change the definition of canonical momentum and 
the motion is identical to that of the original system. 
\end{abstract}


\section{Introduction\label{Introduction}}
In general, Hamiltonian mechanics and Lagrangian mechanics are applied to conservative classical mechanical system or conservative quantum-mechanical system. 
In this paper we attempt to find a sort of relationship between a dissipative classical mechanical system between nonconservative classical mechanical ones, 
then we might apply some methods derived from symplectic geometry to dissipative classical mechanical system.

Some researchers attempt to represent a dissipative system as Hamiltonian formalism or Lagrangian formalism. For instances, about half a century ago, 
Calirola\cite{PCalirola1941},Kana\cite{1948PThPh...3..440K} adopted the Hamiltonian
\begin{equation}
 H_{ck}(q,p)=\frac{1}{2}\left(e^{-2\eta t}p^2+e^{2\eta t}\omega^2 q^2\right),
\label{eq:Hck}
\end{equation}
which leads exactly to the classical equation of motion of a damped harmonic oscillator,
\begin{equation}
 \ddot{x}+2\eta \dot{x}+\omega^2 x^2=0, \ \ \eta>0
\label{eq:1d_d_oscl}
\end{equation}
In this Hamiltonian-description, the canonical momentum is defined as
\[
 p_{ck}=e^{2\eta t}p
\]

In 1940s Morse and Feshbach\cite{book3} gave an example of an artificial Hamiltonian for a damped oscillator based on 
a ``mirror-image'' trick, incorporating a second oscillator with negative friction. The resulting Hamiltonian is 
unphysical: it is unbounded from below and under time reversal the oscillator is transformed into its ``mirror-image''. 
By this arbitrary trick dissipative systems can be handled as though they were conservative. Bateman\cite{PhysRev.38.815} proposed a similar 
approach.
For the system (\ref{eq:1d_d_oscl}), we have
\begin{eqnarray}
 \ddot{x}+2\eta \dot{x}+\omega^2 x^2=0 \ \ (original) \\
 \ddot{y}-2\eta \dot{x}+\omega^2 x^2=0 \ \ (mirror-image).
\label{eq:Bateman}
\end{eqnarray}
Correspondingly, there is Bateman(-Morse-Feshbach) Lagrangian:
\begin{equation}
 L_B(x,\dot{x},y,\dot{y})=\dot{x}\dot{y}+\eta(x\dot{y}-\dot{x}y)-m\omega^2xy
\label{eq:L_batm}
\end{equation}

Rajeev\cite{art7} considered that a large class of dissipative systems can be brought to a canonical form by introducing complex
 coordinates in phase space and a complex-valued Hamiltonian. Rajeev\cite{art7} indicated that Eq.(\ref{eq:1d_d_oscl}) can be brought to diagonal
 form by a linear transformation:
\begin{equation}
 z=A\left[-\mi(p+\eta x)+\omega_1 x\right], \frac{\dif z}{\dif t}=\left[-\gamma+\mi \omega_1\right]z,
\label{eq:complex_transform}
\end{equation}
where 
\begin{equation}
 \omega_1=\sqrt{\omega_2-\gamma^2},
\end{equation}
and the constant $A=1/\sqrt{2\omega_1}$. Then 
\cite{art7} defined the complex-valued function as Hamiltonian
\begin{equation}
 \mathcal{H}=(\omega_1+\mi \eta)zz^*,
\end{equation}
which satisfied
\[
 \frac{\dif z}{\dif t}=\left\lbrace \mathcal{H},z \right\rbrace, \ \ \frac{\dif z^*}{\dif t}=\left\lbrace \mathcal{H},z^* \right\rbrace
\]

By reviewing the works of \cite{PCalirola1941}\cite{1948PThPh...3..440K}\cite{PhysRev.38.815}\cite{book3}\cite{art7}, we can find that they attempt 
to transform a dissipative system into a conservative system entirely and these approaches 
might be suitable for Hamiltonian representation of one-dimensional damped oscillators (weak non-Lagrangian systems) and quantization. 
Because by observing Eq.(\ref{eq:Hck}), Eq.(\ref{eq:Bateman}), Eq.(\ref{eq:L_batm}) and the transformation 
(\ref{eq:complex_transform}), one can find that the damping coefficient is independent of other particles, and \cite{art7} had wrote: '
These complex coordinates are the natural variable(normal modes) of the system. '

In area of quantum mechanics, \cite{PhysRevA.81.022112,Kochan2010219} attempts to quantize dissipative forces in terms of the two form $\Omega$ (
 an analog of $\dif p \wedge \dif q-\dif H \wedge \dif t$), avoiding to obtain Hamiltonian or Lagrangian formulation of non-Lagrangian system.

Marsden \cite{Marsden2007} and other researchers applied the equations as below to the problem of stability of dissipative systems
\begin{eqnarray}
\dot{p}_i&=&- \frac{\partial H}{\partial q_i}+\bm{F}\left( \frac{\partial{r}}{\partial q_i} \right) \nonumber \\
\dot{q}_i&=& \frac{\partial H}{\partial p_i},
\label{eq:eq2}
\end{eqnarray}
where the position vector $r$ depends on the canonical variable $\lbrace q,p \rbrace$, i.e. $r(q,p)$, $H$ denotes Hamiltonian, 
and $\bm{F} (\partial{r}/ \partial{q_i})$ denotes a generalized force in the direction $i$, $i=1,\dots,n$. 
Marsden considered that Eqs.(\ref{eq:eq2}) was composed of a conservative part and a non-conservative part. Eq.(\ref{eq:eq2}) apparently is 
a representation of dissipative mechanical systems in the phase space.
Although one can utilize the approaches discussed in some papers\cite{PCalirola1941}\cite{1948PThPh...3..440K}\cite{PhysRev.38.815}\cite{book3}\cite{art7} to 
convert Eq.({\ref{eq:eq2}}) into a conservative system, one must first change 
the definition of the canonical momentum of the system. If one uses numerical algorithms to solve the Hamiltonian system, the numerical
 solution will lose the physical characteristics of the original system, because the phase flow of the original system is different
 from that of the new system. We need a Hamiltonian system that shares common phase flow or solution with the original system. 
But this demand cannot be satisfied, because it conflicts with Louisville's theorem. Therefore, we would have to attempt to find other relationship between 
dissipative systems and conservative ones.

Based on Eq.(\ref{eq:eq2}), in this paper we will attempt to demonstrate that a dissipative mechanical system shares a single common phase curve with a conservative system. 
In the light of this property, we will propose an approach to substitute a group of conservative systems for a dissipative mechanical system. 
In the following section, we will illustrate the relationship between a dissipative mechanical system and a conservative one.

\section{Relationship between a Dissipative Mechanical System and a Conservative One}

\subsection{A Proposition\label{ObHamiltonian}}
Under general circumstances, the force $\bm{F}$ is a damping force that depends on the variable set $q_1,\cdots,q_n,\dot{q}_1,\cdots,\dot{q}_n$.
$F_i$ denotes the components of the generalized force $\bm{F}$.
\begin{equation}
F_i(q_1,\cdots,q_n,\dot{q}_1,\cdots,\dot{q}_n)=\bm{F}\left( \frac{\partial{r}}{\partial q_i}\right).
\label{eq:inth-1}
\end{equation}
Thus we can reformulate Eq.(\ref{eq:eq2}) as follows:
\begin{eqnarray}
\dot{p}_i&=& - \frac{\partial H}{\partial q_i}
+F_i(q_1,\cdots,q_n,\dot{q}_1,\cdots,\dot{q}_n) \nonumber \\
 \dot{q}_i&=& \frac{\partial H}{\partial p_i}.
\label{eq:inth-2}
\end{eqnarray}
Suppose the Hamiltonian quantity of a conservative system without damping is $\hat{H}$. Thus we
 may write a Hamilton's equation of the conservative system :
\begin{eqnarray}
\dot{p}_i &=& -\frac{\partial {\hat{H}}}{\partial q_i} \nonumber \\
\dot{q}_i &=&\frac{\partial \hat{H}}{\partial p_i}.
\label{eq:inth-3}
\end{eqnarray}
We do not intend to change the definition of momentum in classical mechanics, but we do require that a special solution  
of Eq.(\ref{eq:inth-3}) is the same as that of Eq.(\ref{eq:inth-2}).
We may therefore assume a phase curve $\gamma$ of Eq.(\ref{eq:inth-2}) coincides with 
that of Eq.(\ref{eq:inth-3}). The phase curve $\gamma$ corresponds to an initial condition $q_{i0},p_{i0}$. 
Consequently by comparing Eq.(\ref{eq:inth-2}) and Eq.(\ref{eq:inth-3}), we have
\begin{eqnarray}
\left.\frac{\partial{\hat{H}}}{\partial{q_i}}\right|_{\gamma} &=&
\left.\frac{\partial H}{\partial q_i}\right|_\gamma-\left. F_i(q_1,\cdots,q_n,\dot{q}_1,\cdots,\dot{q}_n)\right|_\gamma \nonumber \\
\left.\frac{\partial{\hat{H}}}{\partial{p_i}}\right|_\gamma&=&
\left.\frac{\partial H}{\partial p_i}\right|_\gamma,
\label{eq:inth-4}
\end{eqnarray}
where $\left.\frac{\partial{\hat{H}}}{\partial{q_i}}\right|_{\gamma},\left.\frac{\partial H}{\partial q_i}\right|_\gamma
,\left.\frac{\partial{\hat{H}}}{\partial{p_i}}\right|_\gamma and \left.\frac{\partial H}{\partial p_i}\right|_\gamma$ denote the values
of these partial derivatives on the phase curve $\gamma$ and $\left.F_i(q_1,\cdots,q_n,\dot{q}_1,\cdots,\dot{q}_n)\right|_\gamma$ denotes 
the value of the force $F_i$ on the phase curve $\gamma$. In classical mechanics the Hamiltonian $H$ of a conservative mechanical 
system is mechanical energy and can be written as:
\begin{equation}
 H=\int_{\gamma}\left(\frac{\partial{H}}{\partial{q_i}}\right)\dif q_i
+\int_{\gamma} \left(\frac{\partial H}{\partial p_i}\right)\dif p_i+const_1,
\label{eq:inth-5}
\end{equation}
where $const_1$ is a constant that depends on the initial condition described above. If $q_i=0,p_i=0$, then $const_1=0$.
The mechanical energy $H$ of the system (\ref{eq:inth-2}) can be evaluated via Eq. (\ref{eq:inth-5}) too.
The Einstein summation convention has been used this section. Thus an attempt has been made to find $\left.\hat{H} \right|_\gamma$ 
through line integral along the phase curve $\gamma$ of the dissipative system
\begin{eqnarray}
\int_{\gamma}\frac{\partial{\hat{H}}}{\partial{q_i}}\dif q_i
&=&\int_{\gamma}\left[\frac{\partial H}{\partial q_i}
-F_i(q_1,\cdots,q_n,\dot{q}_1,\cdots,\dot{q}_n)\right]\ \dif q_i \nonumber \\
 \int_{\gamma} \frac{\partial \hat{H}}{\partial p_i}\dif p_i
&=&\int_{\gamma} \frac{\partial H}{\partial p_i}\dif p_i.
\label{eq:inth-6}
\end{eqnarray}
Analogous to Eq.(\ref{eq:inth-5}), we have
\begin{equation}
 \left.\hat{H}\right|_{\gamma}=\int_{\gamma}\frac{\partial{\hat{H}}}{\partial{q_i}}\dif q_i
+\int_{\gamma} \frac{\partial{\hat{H}}}{\partial p_i}\dif p_i+const_2,
\label{eq:inth-7}
\end{equation}
where $const_2$ is a constant which depends on the initial condition.
Substituting Eq.(\ref{eq:inth-5})(\ref{eq:inth-6}) into Eq.(\ref{eq:inth-7}), we have
\begin{equation}
 \left.\hat{H}\right|_\gamma=H-\int_{\gamma}F_i(q_1,\cdots,q_n,\dot{q}_1,\cdots,\dot{q}_n)\dif q_i+const.
\label{eq:inth-8}
\end{equation}
where $const=const_2-const_1$, and $H=\left.H\right|_{\gamma}$ because $H$ is mechanical energy of the nonconservative system(\ref{eq:inth-2}). 
According to the physical meaning of Hamiltonian, $const_1$, $const_2$ and $const$ are added into Eq.(\ref{eq:inth-5})(\ref{eq:inth-7})(\ref{eq:inth-8})
respectively such that the integral constant vanishes in the Hamiltonian quantity. 
Arnold\cite{Arnold1997} had presented the Newton-Laplace principle of determinacy as, 
'This principle asserts that the state of a mechanical system at any fixed moment of time uniquely
determines all of its (future and past) motion.' In other words, in the phase space the position variable and the velocity variable are 
determined only by the time $t$. Therefore, we can assume that we have already a solution of Eq.(\ref{eq:inth-2})
\begin{eqnarray}
 q_i&=&q_i(t) \nonumber \\
 \dot{q_i}&=&\dot{q_i}(t),
\label{eq:curve}
\end{eqnarray}
where the solution satisfies the initial condition. We can divide the whole time domain into a group of sufficiently small domains and in these domains $q_i$ is monotone, and hence 
we can assume an inverse function $t=t(q_i)$. If $t=t(q_i)$ is substituted into the nonconservative force $\left.F_i\right|_{\gamma}$, we can assume that:
\begin{equation}
\left.F_i(q_1(t(q_i)),\cdots,q_n(t(q_i)),\dot{q}_1(t(q_i)),\cdots,\dot{q}_n(t(q_i)))\right|_{\gamma}= \mathcal{F}_i(q_i),
\label{eq:asumption}
\end{equation}
where $\mathcal{F}_i$ is a function of $q_i$ alone. In Eq.(\ref{eq:asumption}) the function $F_i$ is restricted on the curve $\gamma$, such that a new function
 $\mathcal{F}_i(q_i)$ yields. Thus we have
\begin{eqnarray}
 \int_{\gamma}F_i \dif q_i&=&\int_{q_{i0}}^{q_i}\mathcal{F}_i(q_i) \dif q_i
=W_i(q_i)-W_i(q_{i0}).
\label{eq:inth-9}
\end{eqnarray}
According to Eq.(\ref{eq:inth-9}) the function $\mathcal{F}_i$ is path independent, and therefore $\mathcal{F}_i$ can be regarded as a conservative force. 
For that Eq.(\ref{eq:asumption}) represents an identity map from the nonconservative force $F_i$ on the curve $\gamma$ 
to the conservative force $\mathcal{F}_i$ which is distinct from $F_i$. It must be noted, that Eq.(\ref{eq:asumption}) is tenable only on the phase curve $\gamma$.
 Consequently the function form of $\mathcal{F}_i$ depends on the aforementioned initial condition; from other initial conditions $\mathcal{F}_i$
with different function forms will yield.

According to the physical meaning of Hamiltonian, $const$ is added to Eq.(\ref{eq:inth-8}) such
that the integral constant vanishes in Hamiltonian quantity. Hence $const=-W_i(q_{i0})$.
Substituting Eq.(\ref{eq:inth-9}) and $const=-W_i(q_{i0})$ into Eq.(\ref{eq:inth-8}), we have
\begin{equation}
\left.\hat{H}\right|_{\gamma}=H-W_i(q_i)
\label{eq:inth-10}
\end{equation}
where $-W_i(q_i)$ denotes the potential of the conservative force $\mathcal{F}_i$ and $W_i(q_i)$ is equal to the sum of the work done by the 
nonconservative force $F$ and $const$. In Eq.(\ref{eq:inth-10}) $\hat{H}$ and $H$ are both functions of $q_i,p_i$ and $W_i(q_i)$ 
a function of $q_i$. 
Eq.(\ref{eq:inth-10}) and Eq.(\ref{eq:inth-8}) can be thought of as a map from the total energy of the dissipative system(\ref{eq:inth-2}) to the Hamiltonian
 of the conservative system(\ref{eq:inth-3}). Indeed, $\left.\hat{H}\right|_{\gamma}$ and the total energy differ in the constant $const=-W_i(q_{i0})$. When the
conservative system takes a different initial condition, if one does not change the function form of $\left.\hat{H}\right|_{\gamma}$, one can 
consider  $\left.\hat{H}\right|_{\gamma}$ as a Hamiltonian quantity $\hat{H}$, 
\begin{equation}
 \hat{H}=\left.\hat{H}\right|_{\gamma}=H-W_i(q_i)
\label{eq:inth-10.1}
\end{equation}
and the conservative system(\ref{eq:inth-3}) can be thought of as an entirely 
new conservative system.

Based on the above, the following proposition is made:
\begin{proposition}
For any nonconservative classical mechanical system and any initial condition, there exists a conservative one; the two systems share
one and only one common phase curve;  the value of the Hamiltonian of the conservative system is equal to the sum of the total energy of the nonconservative system on the aforementioned phase curve
and a constant depending on the initial condition.
\label{pro:1}
\end{proposition}
\begin{proof} 
First we must prove the first part of the Proposition \ref{pro:1}, i.e. that a conservative system with Hamiltonian presented by Eq.(\ref{eq:inth-10.1})
shares a common phase curve with the nonconservative system represented by Eq.(\ref{eq:inth-2}). In other words 
the Hamiltonian quantity presented by Eq.(\ref{eq:inth-10.1}) satisfies Eq.(\ref{eq:inth-4}) under the same initial condition. 
Substituting Eq.(\ref{eq:inth-10.1}) into the left side of Eq.(\ref{eq:inth-4}), we have
\begin{eqnarray}
\frac{\partial{\hat{H}(q_i,p_i)}}{\partial {q_i}}&=&\frac{\partial H(q_i,p_i)}{\partial {q_i}}
-\frac{\partial{W_j(q_j)}}{\partial {q_i}} \nonumber\\
\frac{\partial{\hat{H}(q_i,p_i)}}{\partial {p_i}}&=&\frac{\partial H(q_i,p_i)}{\partial {p_i}}
-\frac{\partial{W_j(q_j)}}{\partial {p_i}}.
\label{eq:inth-11}
\end{eqnarray}
It must be noted that although $q_i$ and $p_i$ are considered as distinct variables in Hamilton's mechanics, we can consider $q_i$ and
$\dot{q_i}$ as dependent variables in the process of constructing of $\hat{H}$.
At the trajectory $\gamma$ we have
\begin{eqnarray}
 \frac{\partial{{W_j(q_j) }}}{\partial {q_i}}&=&
\frac{\partial{(\int_{q_{j0}}^{q_j}\mathcal{F}_j(q_j) \dif q_j+W_i(q_{i0}))}}{\partial {q_i}}
=\mathcal{F}_i(q_i) \nonumber \\
\frac{\partial{{W_j(q_j) }}}{\partial {p_i}}&=0,
\label{eq:inth-12}
\end{eqnarray}
where $\mathcal{F}_i(q_i)$ is equal to the damping force $F_i$ on the phase curve $\gamma$. Hence under the initial condition $q_0, p_0$, 
Eq.(\ref{eq:inth-4}) is satisfied. As a result, we can state that the phase curve  of Eq.(\ref{eq:inth-3}) coincides with that of 
Eq.(\ref{eq:inth-2}) under the initial condition;  and $\hat{H}$ represented by Eq.(\ref{eq:inth-10.1}) is 
the Hamiltonian of the conservative system represented by Eq.(\ref{eq:inth-3}).

Then we must prove the second part of Proposition \ref{pro:1}: the uniqueness of the common phase curve.

 We assume that eq.(\ref{eq:inth-3}) shares two common phase curves, $\gamma_1$ and $\gamma_2$, with eq.(\ref{eq:inth-2}). 
Let a point of $\gamma_1$ at the time $t$ be $z_1$, a point of $\gamma_2$ at the time $t$  $z_2$, and 
$g^t$ the Hamiltonian phase flow of eq.(\ref{eq:inth-3}). Suppose a domain $\Omega$ at $t$ which contains only points $z_1$ and $z_2$, and $\Omega$ is 
not only a subset of the phase space of the nonconservative system(\ref{eq:inth-2}) but also that of the phase space of the conservative system(\ref{eq:inth-3}). 
Hence there exists a phase flow $\hat{g}^t$ composed of $\gamma_1$ and $\gamma_2$, and $\hat{g}^t$ is the phase flow of eq.(\ref{eq:inth-2}) restricted by $\Omega$.
According to the following Louisville's theorem\cite{Arnold1978}:
\begin{theorem}
The phase flow of Hamilton's equations preserves volume: for any region $D$ we have
\[
 volume\ of\ g^tD=volume\ of\ D
\]
where $g^t$ is the one-parameter group of transformations of phase space
\[
 g^t:(p(0),q(0))\longmapsto:(p(t),q(t))
\]
\label{Liouville}
\end{theorem}
$g^t$ preserves the volume of $\Omega$. This implies that the phase flow of eq.(\ref{eq:inth-2}) $\hat{g}^t$ preserves the volume of $\Omega$ too. 
 But the system (\ref{eq:inth-2}) is not conservative, which conflicts
with Louisville's theorem; hence only a phase curve of eq.(\ref{eq:inth-3}) coincides with that of eq.(\ref{eq:inth-2}).

\smartqed \qed
\end{proof}
In the next section three examples is given to demonstrate Proposition \ref{pro:1}.
\section{Examples} 
In this section, first two simple analytical examples are given, then a pro forma example is given.
\subsection{One-dimensional Analytical Example}
Consider a special one-dimensional simple mechanical system:
\begin{equation}
 \ddot{x}+c\dot{x}=0,
\label{eq:simp_1d}
\end{equation}
where $c$ is a constant. The exact solution of the equation above is
\begin{equation}
 x=A_1+A_2e^{-ct},
\label{eq:sol_1d}
\end{equation}
where $A_1,A_2$ are constants. From the equation above, we derived the velocity:
\begin{equation}
 \dot{x}=-cA_2e^{-ct}.
\label{eq:sol_1dv}
\end{equation}
From the initial condition $x_0,\dot{x}_0$, we find $A_1=x_0+\dot{x}_0/c, A_2=-\dot{x}_0/c$. From Eq.(\ref{eq:sol_1d})
\begin{equation}
 t=-\frac{1}{c}\ln\frac{x-A_1}{A_2}
\label{eq:tfunc}
\end{equation}
Substituting the equation above into Eq.(\ref{eq:sol_1dv}), such we have
\begin{equation}
 \dot{x}=-c(x-A_1)=-c(x-A_1)
\label{eq:dx}
\end{equation}
The dissipative force $F$ in the dissipative system (\ref{eq:simp_1d}) is
\begin{equation}
 F=c\dot{x}.
\label{eq:F}
\end{equation}
Substituting Eq.(\ref{eq:dx}) into Eq.(\ref{eq:F}), such we have the conservative force $\mathcal{F}$
\begin{equation}
 \mathcal{F}=-c^2(x-A_1);
\label{eq:mF}
\end{equation}
Clearly the conservative force $\mathcal{F}$ depends on the initial condition of the dissipative system (\ref{eq:simp_1d}), in other words 
an initial condition determine a conservative force. Consequently a new conservative system yields
\begin{equation}
 \ddot{x}+\mathcal{F}=0\rightarrow \ddot{x}-c^2(x-A_1)=0.
\label{eq:1d_eq_consys}
\end{equation}
The stiffness coefficient of the equation above must be negative. One can readily verify that the particular solution (\ref{eq:sol_1d}) 
of the dissipative system can satisfy the conservative one (\ref{eq:1d_eq_consys}). This point agrees with Proposition (\ref{pro:1}).

The potential of the conservative system（\ref{eq:1d_eq_consys}）is 
\[
 V=\int_0^x \left[ -c^2(x-A_1) \right] =-\frac{c^2}{2}x^2+c^2A_1x 
\]
If $t\rightarrow\infty$，$x\rightarrow A_1$ and $\dot{x}\rightarrow 0$. This implies that the kinetic energy of the corresponding conservative system 
would tend to $0$ and the potential a constant $C^2A_1^2/2$ which is equal to the energy loss of the original system.
Both the mechanical energy of the conservative system (\ref{eq:1d_eq_consys}) at initial instance and $t\rightarrow\infty$ are 
$c^2A_1^2/2$.

\subsection{Two-dimensional Analytical Example}
Let us consider a special two-dimensional mechanical system
\begin{eqnarray}
\ddot{x}+\dot{x}-\dot{y}&=&0\nonumber\\
\ddot{y}-\dot{x}+\dot{y}&=&0.
\label{eq:simple_2d}
\end{eqnarray}
The exact solution of the equation above with initial initial condition $x_0,y_0,\dot{x}_0,\dot{y}_0$ is
\begin{eqnarray}
x(t) &=&-\frac{\dot{y}_0-\dot{x}_0-4x_0}{4}+\frac{{e}^{-2t}( \dot{y}_0-\dot{x}_0) }{4}+\frac{t( \dot{y}_0) }{2}+\frac{t(\dot{x}_0)}{2}\nonumber\\
y(t) &=&\frac{\dot{y}_0-\dot{x}_0+4y_0}{4}-\frac{{e}^{-2t}( \dot{y}_0-\dot{x}_0) }{4}+\frac{t( \dot{y}_0) }{2}+\frac{t( \dot{x}_0)}{2}
\label{eq:sol_2d}
\end{eqnarray}
For convenience to obtain $t=t(x),t=t(y)$，let $\dot{x}_0+\dot{y}_0=0$, then simplify the particular solution above to
\begin{eqnarray}
 x(t) =-\frac{\dot{y}_0-\dot{x}_0-4x_0}{4}+\frac{{e}^{-2t}( \dot{y}_0-\dot{x}_0) }{4}\nonumber\\
y(t) =\frac{\dot{y}_0-\dot{x}_0+4y_0}{4}-\frac{{e}^{-2t}( \dot{y}_0-\dot{x}_0) }{4},
\label{eq:s_sol_2d}
\end{eqnarray}

From the equation above, we derived the velocity:
\begin{eqnarray}
\dot{x} &=&-\frac{{e}^{-2t}( \dot{y}_0-\dot{x}_0) }{2} ,\label{eq:sol_2d_vx}\\
\dot{y} &=&\frac{{e}^{-2t}( \dot{y}_0-\dot{x}_0) }{2}
\label{eq:sol_2d_vy}
\end{eqnarray}
Let the phase curve be denoted as $\gamma$.
From Eq.(\ref{eq:s_sol_2d}), we obtain the inverse functions
\begin{eqnarray}
t&=&-\frac{1}{2}\ln\left[\frac{4}{\dot{y}_0-\dot{x}_0}(x-x_0)+1 \right] \label{eq:tfunc_x}\\
t&=&-\frac{1}{2}\ln\left[-\frac{4}{\dot{y}_0-\dot{x}_0}(y-y_0)+1 \right]
\label{eq:tfunc_y}
\end{eqnarray}
Substituting Eq.(\ref{eq:tfunc_x})(\ref{eq:tfunc_y}) into Eq.(\ref{eq:sol_2d_vx}), we have the map at $\gamma$ from $x,y$ to $\dot{x}$:
\begin{eqnarray}
 \dot{x}(x)&=&-2x-\frac{\dot{y}_0-\dot{x}_0}{2}+2x_0\label{eq:dx1}\\
 \dot{x}(y)&=&2y-\frac{\dot{y}_0-\dot{x}_0}{2}-2y_0\label{eq:dx2}
\end{eqnarray}
Substituting Eq.(\ref{eq:tfunc_x})(\ref{eq:tfunc_y}) into Eq.(\ref{eq:sol_2d_vy}), we have the map at $\gamma$ from $x,y$ to $\dot{y}$:
\begin{eqnarray}
 \dot{y}(y)&=&-2y+\frac{\dot{y}_0-\dot{x}_0}{2}+2y_0\label{eq:dy1}\\
\dot{y}(x)&=&2x+\frac{\dot{y}_0-\dot{x}_0}{2}-2x_0 \label{eq:dy2}
\end{eqnarray}
The components of nonconservative $\bm{F}$ in the system (\ref{eq:simple_2d}) are
\begin{eqnarray}
 F_1&=&\dot{x}-\dot{y} \label{eq:F1}\\
 F_2&=&-\dot{x}+\dot{y}\label{eq:F2}
\end{eqnarray}
Substituting Eq.(\ref{eq:dx1})(\ref{eq:dy2}) into $F_1$(\ref{eq:F1}), then take the quantity as the first component the conservative force
$\mathcal{F}$:
\begin{equation}
 \mathcal{F}_1(x)=-4x-(\dot{y}_0-\dot{x}_0)+4x_0.
\label{eq:mF1}
\end{equation}
Substituting Eq.(\ref{eq:dx2})(\ref{eq:dy1}) into $F_2$(\ref{eq:F2}), then take the quantity as the second component the conservative force
$\mathcal{F}$:
\begin{equation}
 \mathcal{F}_2(y)=-4y+(\dot{y}_0-\dot{x}_0)+4y_0
\label{eq:mF2}
\end{equation}
Since $\partial \mathcal{F}_1/\partial y=\partial \mathcal{F}_2/\partial x=0$, $\mathcal{F}$ must be conservative. 
Consequently we obtain a new conservative system:
\begin{eqnarray}
\ddot{x}&=&-\mathcal{F}_1 \nonumber\\
                  &=&4x+(\dot{y}_0-\dot{x}_0)-4x_0\nonumber\\
\ddot{y}&=&-\mathcal{F}_2 \nonumber\\
&=&4y-(\dot{y}_0-\dot{x}_0)-4y_0.
\label{eq:2d_eq_consys}
\end{eqnarray}
We can readily prove that the particular solution (\ref{eq:s_sol_2d}) can satisfy Eq.(\ref{eq:2d_eq_consys}) too.
In this case, this point agrees with Proposition \ref{pro:1} too.

\subsection{A Formell Example in Vibration Mechanics\label{example}}

Take an $n$-dimensional oscillator with damping as an example, the governing equation of which is as below:
\begin{equation}
\ddot{\bm{q}}+\tensor{C}\dot{\bm{q}}+\tensor{K}\bm{q}=0,
\label{eq:ex2-1}
\end{equation}
where $\bm{q}=\left[q_1,\dots,q_n \right] ^T$, superscript $T$ denotes a matrix transpose,
\[
 \tensor{C}=
\left[
\begin{array}{ccc}
C_{11}&\dots&C_{1n}\\
\vdots&\ddots&\vdots\\
C_{n1}&\dots&C_{nn}
\end{array}
\right],
\tensor{K}=
\left[
\begin{array}{ccc}
K_{11}&\dots&K_{12}\\
\vdots&\ddots&\vdots\\
K_{21}&\dots&K_{22}
\end{array}
\right]
\], 
and $C_{ij}$ and $K_{ij}$ are constants.

It is complicated to solve Eq.(\ref{eq:ex2-1}). If Eq.(\ref{eq:ex2-1}) is higher dimensional, it is almost impossible to solve Eq.(\ref{eq:ex2-1}) analytically. 
Therefore we assume that a solution exists already.
\begin{equation}
 \bm{q}=\bm{q}(t)=\left[q_1(t),\dots,q_n(t)\right].
\label{eq:ex2-2}
\end{equation}
Suppose a group of inverse functions
\begin{equation}
 t=t(q_1),\dots , t=t(q_n).
\label{eq:ex2-3}
\end{equation}
As in Eq.(\ref{eq:asumption}) we can consider that the damping forces are equal to some conservative force under an initial condition
\begin{equation}
\begin{array}{ccc}
c_{11}\dot{q}_1=\varrho_{11}(q_1)&\dots&c_{1n}\dot{q}_n=\varrho_{1n}(q_1)\\
\vdots&\ddots&\vdots\\
c_{n1}\dot{q}_1=\varrho_{21}(q_n)&\dots&c_{nn}\dot{q}_n=\varrho_{nn}(q_n),
\end{array}
\label{eq:ex2-4}
\end{equation}
where $\varrho_{ij}(q_i)$ is a function of $q_i$. 
For convenience, these conservative forces can be defined as functions which are analogous to elastic restoring forces: 
\begin{equation}
\begin{array}{ccc}
\varrho_{11}(q_1)=\kappa_{11}(q_1)q_1&\dots&\varrho_{1n}(q_1)=\kappa_{1n}(q_1)q_1\\
\vdots&\ddots&\vdots\\
\varrho_{n1}(q_1)=\kappa_{n1}(q_n)q_n&\dots &\varrho_{nn}(q_n)=\kappa_{nn}(q_n)q_n,
\end{array}
\label{eq:ex2-5}
\end{equation}
where $\kappa_{ij}(q_i)$ is a function of $q_i$.
An equivalent stiffness matrix $\tensor{\tilde{K}}$ is obtained, which is a diagonal matrix
\begin{equation}
\tensor{\tilde{K}}_{ii}=\sum_{l=1}^n \kappa_{il}(q_l).
\label{eq:ex2-5a}
\end{equation}
Consequently an $n$-dimensional conservative system is obtained
\begin{equation}
 \bm{\ddot{q}}+(\tensor{K}+\tensor{\tilde{K}})\bm{q}=0
\label{eq:ex2-6}
\end{equation}
which shares a common phase curve with the $n$-dimensional damping system(\ref{eq:ex2-1}).
The Hamiltonian of Eqs.(\ref{eq:ex2-6}) is
\begin{equation}
 \hat{H}=\frac{1}{2}\bm{p}^T\bm{p}+\frac{1}{2}\bm{q}^T \tensor{K}\bm{q}+
\int_{\bm{0}}^{\bm{q}} (\tilde{\tensor{K}}\bm{q})^T \dif \bm{q},
\label{eq:ex2-7}
\end{equation}
where $\bm{0}$ is a zero vector, $\bm{p}=\dot{\bm{q}}$. $\hat{H}$ in Eq.(\ref{eq:ex2-7}) is the mechanical energy of the conservative 
system(\ref{eq:ex2-6}), because $\int_{\bm{0}}^{\bm{q}} (\tilde{\tensor{K}}\bm{q})^T \dif \bm{q}$
 is a potential function such that $\hat{H}$ doest not 
depend on any path.

\subsection{Discussion}

Based on the above, we can outline the relationship between a dissipative mechanical system and a group of conservative systems by means of Fig. \ref{fig1}. 
The relationship can be stated from two perspectives:
\begin{figure}
\begin{center}
\includegraphics[totalheight=2.5in]{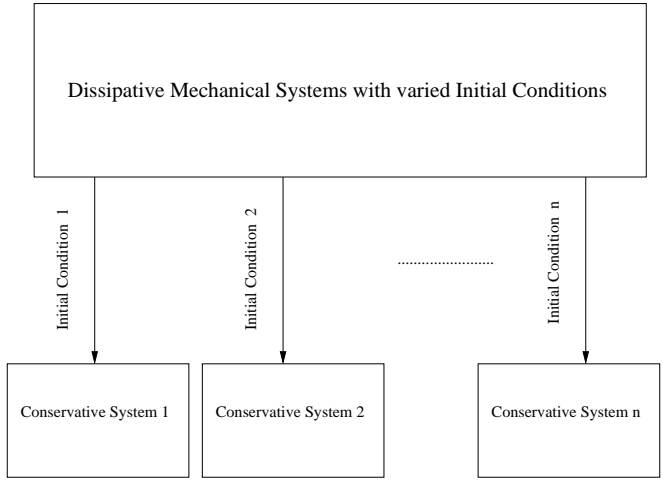}
\caption{A Dissipative Mechanical System and Conservative Systems}
\label{fig1}
\end{center}
\end{figure}

If one explains the relationship from a geometrical perspective, one can obtain Proposition \ref{pro:1}. In this paper 
the conservative systems (\ref{eq:inth-3}) and (\ref{eq:ex2-6}) are called the substituting systems. Although a 
substituting system shares a common phase curve with the original system, under other initial conditions the 
substituting system exhibits different phase curves. Therefore the phase flow of the substituting system differs 
from that of the original system, it follows that the substituting systems is not equal to 
the original system. According to Louisville's theorem (\ref{Liouville}), the phase flow of the original 
dissipative system Eq.(\ref{eq:inth-2}) certainly does not preserve its phase volume, but the phase flow of the 
substituting conservative Eq.(\ref{eq:inth-3}) does.

One also could explain the relationship from a mechanical perspective. 
It is known that there are non-conservative forces in a nonconservative system. The total energy of the nonconservative system consists of 
the work done by nonconservative forces. Hence the function form of the total energy depends on a phase curve i.e. under an initial condition. 
If one constrains the total energy function to a phase curve $\gamma$, the total energy function can be converted into a function of $q,p$. 
One take $\hat{H}$ consisting of this new function and a constant as a Hamiltonian quantity, such that a Hamilton's system (i.e., a 
conservative system) is obtained. Under the initial condition mentioned above, the solution curve of the 
conservative system is the same as that of the original nonconservative system; under other initial conditions 
the solution curve of the conservative is different from that of the original nonconservative system. 
Since one defines the forces(\ref{eq:asumption},\ref{eq:ex2-4},\ref{eq:ex2-5},\ref{eq:ex2-5a}) in the new system,
the Hamiltonian quantity of the conservative can be thought of as the mechanical energy of the new conservative system as Eq.(\ref{eq:ex2-7}).

One might doubt that the orbit of a dissipative dynamical system must be asymptotic, can the asymptotic orbit coincide with one of a conservative mechanical system. 
In some literature\cite{Sunyishui_book_e_2008}, a conservative system defined a system with the behavior of the preservation of phase volume. 
Hasselblatt\cite{Hasselblatt_Katok2003} had explained the question: 'A key to understanding this difference is given by a property that is not directly
observed by looking at individual orbits but by considering the evolution of large sets of initial conditions simultaneously, the preservation of phase volume.'
This point agrees with the second part of the proof of the Proposition \ref{pro:1}.

The Hamiltonians of the new conservative systems in general are not analytically integrable, unless
 the original mechanical system is integrable. The reason is that the work done by damping force
depends on the phase curve. If the system is integrable, then the
phase curve can be explicitly written out, the system has an analytical
solution, and therefore the work done by damping force can be
explicitly integrated. Subsequently, the Hamiltonian $\hat{H}$ can be explicitly expressed.
Most systems do not have an analytical solution. Despite this, the Hamilton quantity, coordinates and momentum must satisfy
 Eq.(\ref{eq:inth-3}) under a certain initial condition.  Why had Klein\cite{Klein1928} written, ''Physicists can make use of these theories only very little,
an engineers nothing at all''?  The answer: when one is seeking an analytical solution to a classical mechanics problem by 
utilizing Hamiltonian formalism, in fact one must inevitably convert the problem back to Newtonian formalism. This means that an explicit form of  Hamiltonian quantity is not necessary for classical mechanics.
What is important is the relationship between $q,p$ and the Hamiltonian quantity embodied in the Hamilton's Equation.

\section{conclusions}
We can conclude that a dissipative mechanical system has such
properties: for any nonconservative classical mechanical system and any initial condition, 
there exists a conservative one, the two systems share
one and only one common phase curve;  the Hamiltonian of the conservative 
system is the sum of the total energy of the nonconservative system on the aforementioned phase curve
and a constant depending on the initial condition.  We can further conclude,
that a dissipative problem can be reformulated as an infinite number
of non-dissipative problems, one corresponding to each phase curve
of the dissipative problem. One can avoid having to change the
definition of the canonical momentum in the Hamilton formalism,
because under a certain initial condition the motion of one of the
group of conservative systems is the same as the original dissipative
system.

\bibliography{mybib.bib}

\begin{thebibliography}{10}
\providecommand{\url}[1]{{#1}}
\providecommand{\urlprefix}{URL }
\expandafter\ifx\csname urlstyle\endcsname\relax
  \providecommand{\doi}[1]{DOI~\discretionary{}{}{}#1}\else
  \providecommand{\doi}{DOI~\discretionary{}{}{}\begingroup
  \urlstyle{rm}\Url}\fi

\bibitem{Arnold1978}
Arnold., V.I.: Mathematical Methods of classical Mechanics, second edition.
\newblock Springer-Verlag, Berlin (1978)

\bibitem{Arnold1997}
Arnold., V.I.: Mathematical aspects of classical and celestial mechanics.
\newblock Springer-Verlag, Berlin (1997)

\bibitem{PhysRev.38.815}
Bateman, H.: On dissipative systems and related variational principles.
\newblock Phys. Rev. \textbf{38}(4), 815--819 (1931).
\newblock \doi{10.1103/PhysRev.38.815}

\bibitem{PCalirola1941}
Caldirola, P.: Forze non conservative della meccanica quantistica.
\newblock Nuovo Cim \textbf{18}, 393--400 (1941)

\bibitem{Klein1928}
F.Klein: Entwickelung der Mathematik im 19 Jahrhundert.
\newblock Teubner (1928)

\bibitem{Hasselblatt_Katok2003}
Hasselblatt, B., Katok, A.: A FIRST COURSE IN DYNAMICS with a Panorama of
  Recent Developments.
\newblock AMBRIDGE UNIVERSITY PRESS (2003)

\bibitem{1948PThPh...3..440K}
{Kanai}, E.: {On the Quantization of the Dissipative Systems}.
\newblock Progress of Theoretical Physics \textbf{3}, 440--442 (1948)

\bibitem{PhysRevA.81.022112}
Kochan, D.: Functional integral for non-lagrangian systems.
\newblock Phys. Rev. A \textbf{81}(2), 022,112 (2010).
\newblock \doi{10.1103/PhysRevA.81.022112}

\bibitem{Kochan2010219}
Kochan, D.: How to quantize forces (?): An academic essay on how the strings
  could enter classical mechanics.
\newblock Journal of Geometry and Physics \textbf{60}(2), 219 -- 229 (2010).
\newblock \doi{DOI: 10.1016/j.geomphys.2009.09.014}.
\newblock
  \urlprefix\url{http://www.sciencedirect.com/science/article/B6TJ8-4XDCHN8-1/%
2/ed88d72d5e7e026557c477b9416a744e}

\bibitem{Marsden2007}
{Krechetnikov}, R., {Marsden}, J.E.: {Dissipation-induced instabilities in
  finite dimensions}.
\newblock Reviews of Modern Physics \textbf{79}, 519--553 (2007).
\newblock \doi{10.1103/RevModPhys.79.519}

\bibitem{book3}
P.Morse, Feshbach, H.: Methods of Theoretical Physics.
\newblock McGraw-Hill, New York (1953)

\bibitem{art7}
Rajeev, S.: A canonical formulation of dissipative mechanics using
  complex-valuedhamiltonians.
\newblock ANNALS of PHYSICS \textbf{322}(3), 1541--1555 (2007)

\bibitem{Sunyishui_book_e_2008}
Sun, Y., Zhou, Y.: Introduction to Modern Celestial Mechanics.
\newblock Higher Education Press (2008)

\end{thebibliography}

\end{document}